\documentclass[12pt]{iopart}
\usepackage{iopams}
\usepackage{amsopn}
\usepackage{setstack}
\usepackage{bbm}
\usepackage[german,english]{babel}  
\usepackage[T1]{fontenc}
\usepackage{url} 
\usepackage{pst-uml}
\usepackage{dsfont}
\usepackage{amsthm}
\newcommand{\ud}{\mathrm{d}}
\newcommand{\ui}{\mathrm{i}}
\newcommand{\ue}{\mathrm{e}}

\newcommand{\rz}{{\mathbb R}}
\newcommand{\nz}{{\mathbb N}}
\newcommand{\kz}{{\mathbb C}}
\newcommand{\Th}{\Theta}
\newcommand{\The}{\Theta}
\newcommand{\dz}{d}
\newcommand{\Thd}{\Theta_{\Delta^{\frac{1}{2}}}}
\newcommand{\be}{\begin{eqnarray}}
\newcommand{\ee}{\end{eqnarray}}
\newcommand{\eq}{\eqalign}
\newcommand{\Tho}{\The^\mathrm{Osc.}}
\newcommand{\Thw}{\The^W}
\newcommand{\al}{\alpha}
\newcommand{\hata}{\tilde}
\newcommand{\lf}{\mathfrak{l}}
\newcommand{\lfp}{\lf_{p,n}}
\DeclareMathOperator{\re}{Re}

\DeclareMathOperator{\orr}{o}
\DeclareMathOperator{\Ei}{Ei}
\DeclareMathOperator{\Ee}{E_1}

\DeclareMathOperator{\Es}{E^{\ast}}
\DeclareMathOperator{\Dz}{D}
\DeclareMathOperator{\Lz}{L}
\newtheorem{theorem}{Theorem}[section]

\newtheorem{cor}[theorem]{Corollary}
\newtheorem{defn}[theorem]{Definition}

\begin{document}
\title[A new proof of the Vorono\"i summation formula]{A new proof of the Vorono\"i summation formula}
\author{Sebastian Egger né Endres}
\address{Institut f{\"u}r Theoretische Physik, Universit{\"a}t Ulm,\newline Albert-Einstein-Allee 11, 89081 Ulm, Germany}
\ead{sebastian.endres@uni-ulm.de}
\author{Frank Steiner}
\address{Institut f{\"u}r Theoretische Physik, Universit{\"a}t Ulm,\newline Albert-Einstein-Allee 11, 89081 Ulm, Germany \\ \vspace{2mm} and  \\ \vspace{2mm}
Centre de Recherche Astrophysique de Lyon, Universit\'{e} Lyon 1, CNRS UMR 5574, Observatoire de Lyon, 9 avenue Charles Andr\'{e}, 69230 Saint-Genis Laval, France}
\ead{frank.steiner@uni-ulm.de}
\begin{abstract}
We present a short alternative proof of the Vorono\"i summation formula which plays an important role in Dirichlet's divisor problem and has recently found an application in physics as a trace formula for a Schrödinger operator on a non-compact quantum graph $\mathfrak{G}$ [S. Egger né Endres and F. Steiner, J. Phys.\,A: Math. Theor. $\boldsymbol{44}$ ($2011$) $185202$ ($44$pp)]. As a byproduct we give a new proof of a non-trivial identity for a particular Lambert series which involves the divisor function $d(n)$ and is identical with the trace of the Euclidean wave group of the Laplacian on the infinite graph $\mathfrak{G}$.
\end{abstract}
\pacs{03.65.Ca, 03.65.Db}
\maketitle

\section{Introduction: The Vorono\"i summation formula and its connection to an infinite quantum graph}
\label{1}

The Vorono\"i summation formula provides a connection between a sum involving the divisor function $d(n)$ weighted with a given function $f(n)$ and a sum involving again the divisor function weighted with a function $F(n)$ which arises from a certain Bessel transformation of the function $f$.  The divisor function is defined as the number of divisors of $n$, unity and $n$ itself included, i.e. 
\begin{equation}
\label{8}
\dz(k):=\#\left\{(n,m);\quad nm=k\right\}, \quad n,m,k\in\nz.
\end{equation}
Note that the divisor function $d(n)$, $d(1)=1$, $d(2)=2$, $d(3)=2$, $d(4)=3$, $d(5)=2$, $d(6)=4$, $\ldots$, with $d(p)=2$ for $p$ prime, is a very irregular function with asymptotic behaviour 
\begin{equation}
\label{9a}
d(n)=\Or\left(n^{\epsilon}\right),\quad n\rightarrow\infty \quad \mbox{for all} \quad \epsilon>0.
\end{equation}
In 1849 Dirichlet proved \cite{Dirichlet:1849} the following  asymptotic formula 
\be
\label{9aa}
D(x):=\sum\limits_{n=1}^{x}d(n)=x\ln x+(2\gamma-1)x+\Delta(x),\quad x\rightarrow\infty
\ee
with
\be
\label{9ab}
\Delta(x)=\Or\left(\sqrt{x}\right),\quad x\rightarrow\infty,
\ee
where $\gamma$ is Euler's constant. The famous Dirichlet's divisor problem is that of determining as precisely as possible the maximum order of the error term $\Delta(x)$. In 1903 Vorono\"i \cite{Voronoi:1903} was able to improve Dirichlet's result by proving 
\be
\label{9ac}
\Delta(x)=\Or\left(x^{\frac{1}{3}}\ln x\right),\quad x\rightarrow\infty.
\ee
Vorono\"i based his prove on a new summation formula carrying now its name, see \cite{Voronoi:1904}. The estimate on $\Delta(x)$ was later improved, see e.g. \cite{Vandercorpu1922t:,Huxley:2003}. In \cite{Hardy:1917} it was proved that
\be
\label{9ad}
\Delta(x)=\Or\left(x^{\al}\right),\quad x\rightarrow\infty\quad \mbox{with}\quad \al>\frac{1}{4},
\ee
and it is not unlikely that (Hardy's conjecture)
\be
\label{9aba}
\Delta(x)=\Or\left(x^{\frac{1}{4}+\epsilon}\right),\quad x\rightarrow\infty,
\ee
for all positive values of $\epsilon$, but the exact order of $\Delta(x)$ is still unknown. In various articles (e.g. \cite{Rogosinski:1922,Oppenheim:1927,Landau:1927,Koshliakov:1928,Dixon:1931,Wilton:1932,Ferrar:1935,Ferrar:1937,Dixon:1937,Chandrasekharan:1968,Berndt:1972,Hejhal:1979}) the authors investigate the Dirichlet divisor problem respectively the Vorono\"i summation formula and specify proper function spaces for which the Vorono\"i summation formula is valid. See also \cite{Ferrar:1935,Ferrar:1937} for a generalization to a broader class of summation formulae. The result of Hejhal \cite{Hejhal:1979} states:
\pagebreak
\begin{theorem}[Hejhal:1979, \cite{Hejhal:1979}]
\label{13a}
If $f$ is two times continuously differentiable and possesses compact support ($f\in C_0^2(\rz)$) then the following formula is valid
\begin{eqnarray}
\label{z2}
%
\sum\limits_{n=1}^{\infty}d(n) f(n) &=&\int\limits_0^{\infty}\left(\ln k+2\gamma\right)f(k)\ud k+\frac{f(0)}{4}\nonumber\\ 
&\phantom{=}& \hspace{-0.9cm}+2\pi\sum\limits_{n=1}^{\infty}d(n)\int\limits_0^{\infty}\left[\frac{2}{\pi}K_0\left(4\pi\sqrt{nk}\right)-Y_0\left(4\pi\sqrt{nk}\right)\right]f(k)\ud k, \label{13}
\end{eqnarray}
where the functions $K_0(x)$ and $Y_0(x)$ are Bessel functions, i.e. the McDonald function respectively the Neumann function (see e.g. $\cite[pp.\,65,66]{Magnus:1966})$ and $\gamma$ denotes the Euler constant.
\end{theorem}
Recently, the summation formula \eref{z2} found an interesting application in physics since it was shown in \cite{SteinerEndres:2008,Egger:2011} that it plays the role of an exact trace formula for a certain infinite quantum graph $\mathfrak{G}$. The infinite quantum graph $\mathfrak{G}$ introduced in \cite{SteinerEndres:2008} and investigated in detail in \cite{Egger:2011}  describes the quantum dynamics of a single spinless particle with mass $m=\frac{1}{2}$ ($\hbar=1$) moving on an infinite chain of consecutive edges $\{e_n\}_{n=1}^{\infty}$ with corresponding lengths $l_n=\frac{\pi}{n}$, $n\in\nz$, where two adjacent edges $e_l$ and $e_{l+1}$ are linked by a vertex $v_l$, $l\in\nz$. The associated Hilbert space is equivalent with the set of square integrable functions $\Lz^2[0,\infty)$ on the semi axis $\rz_{\geq0}$. The Schrödinger operator $\left(-\Delta,\Dz(\Delta)\right)$ consists of (minus) the Laplace operator $-\Delta$ acting on each edge separately. The assigned domain of definition $\Dz(\Delta)$  is a ``Sobolev-like'' function space (see \cite{SteinerEndres:2008,Egger:2011}) and is characterized by Dirichlet boundary conditions at the vertices. Therefore, the corresponding classical system is integrable  and the lengths of the corresponding classical primitive periodic orbits are given by $\lfp=\frac{2\pi}{n}$, $n\in\nz$. The spectrum of the infinite quantum graph $\mathfrak{G}$ is purely discrete, and the $n^{th}$ energy level of $\mathfrak{G}$ is given by $E_n=n^2$, $n\in\nz$, with the corresponding multiplicity precisely equal to the divisor function $d(n)$. Thus, the corresponding wavenumber counting function 
\be
\label{9y}
\hspace{-0.5cm} N(k):=\#\left\{0<k_n\leq k; \ k_n:=\sqrt{E_n}, \ E_n \ \mbox{eigenvalue of} \ \left(-\Delta,\Dz(\Delta)\right)\right\} 
\ee
coincides with $D(k)$ in \eref{9aa} and possesses the non-standard Weyl asymptotics $N(k)\sim k\ln k$, $k\rightarrow\infty$. Therefore, it coincides after a rescaling of the edge lengths by a factor $\frac{1}{2\pi}$ with the leading asymptotic term of the counting function for the non-trivial zeros of the Riemann zeta function $\zeta(s)$. This is interesting in the context of the search of the highly desired hypothetical Hilbert-Polya operator (see e.g. \cite{Berry:1999,Bogomolny:2007,Steiner:2009,Endres:2010}), i.e. a self adjoint operator whose eigenvalues or wavenumbers coincide with the non-trivial Riemann zeros and whose existence would prove the Riemann hypothesis.
 
One can show \cite{Egger:2011} that the Dirichlet quantum graph $\mathfrak{G}$ arises as a limit operator in the strong resolvent sense in the limit $\kappa$ goes to infinity (\cite[p.\,455]{Kato:1980} and \cite{Egger:2011}) of the Schrödinger operator 
\be
\label{z1}
\hspace{-1.5cm}H=-\Delta+\kappa\sum\limits_{n=1}^{\infty}\delta\left(x-x_n\right),\quad\kappa\geq0,\quad x\geq0,\quad x_n:=\pi\sum\limits_{m=1}^{n}\frac{1}{m},\quad n\in\nz, 
\ee
acting in $\Lz^2[0,\infty)$. The Hamiltonian \eref{z1} describes a single spinless particle with mass $m=\frac{1}{2}$ ($\hbar=1$) on the semi axis $\rz_{\geq0}$ subjected to $\delta$-potentials with strength $\kappa\geq0$ located at the points $x_n$, $n\in\nz$. Similarly, as for the Dirichlet quantum graph $\mathfrak{G}$ one can show that the spectrum of the Schrödinger operator \eref{z1} is purely discrete for $\kappa>0$ \cite{Egger:2011}, which is due to the monotonic decrease of the edge lengths $l_n$, $n\in\nz$, to zero.

%
%
%
%

Given a function $f(k)$ satisfying the conditions of theorem \ref{13a}, let us consider the operator function $f\left(\sqrt{-\Delta}\right)$ (defined by the functional calculus), where $(-\Delta,\Dz(\Delta))$ denotes the negative Laplacian acting on our infinite graph $\mathfrak{G}$. We then obtain the trace-identity
\be
\label{z109}
\tr f\left(\sqrt{-\Delta}\right)=\sum\limits_{n=1}^{\infty}d(n)f\left(k_n\right),
\ee
which clearly shows that the Vorono\"i summation formula \eref{z2} can be interpreted as a trace formula for the quantum graph $\mathfrak{G}$. Indeed, the trace \eref{z109}, i.e. the l.h.s. of \eref{z2} is nothing else than a sum over the quantal spectrum of $\mathfrak{G}$ by identifying $k_n=n$ as the wavenumbers and $d(n)$ as the spectral multiplicities. Furthermore, the first two terms on the r.h.s of \eref{z2} play the role of a ``Weyl term'' which is connected to the two leading asymptotic terms of $D(k)$ in \eref{9aa} respectively of the spectral counting function $N(k)$ in \eref{9y} (see \cite{SteinerEndres:2008,Egger:2011}). The last term on the r.h.s. of the Vorono\"i summation formula \eref{z2} presents the ``periodic orbit term'' of the trace formula and invokes the geometric properties in terms of the lengths $\lfp:=\frac{2\pi}{n}$, $n\in\nz$, of the classical primitive periodic orbits ($\lfp$ is the length of the $n^{th}$ classical primitive periodic orbit $p$) of the corresponding classical dynamics of the infinite graph $\mathfrak{G}$.

The typical splitting of the classical part of a trace formula for physical systems in a Weyl term and a periodic orbit term occurs also in the case of the Selberg trace formula \cite{Selberg:1956,Steiner:1987,SteinerAurich:1988}, the semiclassical Gutzwiller trace formula \cite{Gutzwiller:1971,Gutzwiller:1990,Berry:1999,Sieber:1990}, the trace formula for compact quantum graphs \cite{KottosSmilansky:1998,Kurasov:2005,BE:2008} and in the flat torus model \cite{Steiner:2009}.

The main very interesting difference to the usual physical trace formulae as for instance the Selberg trace formula valid for a free particle moving on a Riemannian manifold of constant negative curvature (genus $g\geq2$, [classical chaotic system]) \cite{Selberg:1956,SteinerAurich:1988,Steiner:2009}, the semiclassical Gutzwiller trace formula ($\hbar\rightarrow0$) for classical chaotic Hamiltonian systems \cite{Gutzwiller:1971,Gutzwiller:1990,Sieber:1990,Berry:1999}, the trace formula for the flat torus model \cite{Steiner:2009}, or the trace formula for compact quantum graphs \cite{KottosSmilansky:1998,Kurasov:2005,BE:2008} is that the lengths $\lfp=\frac{2\pi}{n}$ of the classical periodic orbits appear in the denominator of the argument in the periodic orbit term of \eref{z2}. One physical reason of this non-standard property is that for this special system the lengths of the classical primitive periodic orbits $\lfp=\frac{2\pi}{n}$, $n\in\nz$, possess zero as the only accumulation point and are bounded from above (by $2\pi$) in contrast to the generic case for classical chaotic systems where the lengths of the classical primitive periodic orbits tend to infinity and are bounded from below. 

In the semiclassical analysis of quantum mechanical quantities beyond the Weyl regime, in particular in the periodic-orbit theory for spectral statistics \cite{Berry:1986,Aurich:94,Berry:1999,Gnutzman:2006,Bogomolny:2007}, the relevant terms for this analysis are constituted by primitive periodic orbits with lengths in the neighbourhood of the accumulation point (infinitely many). This means for the Dirichlet quantum graph that the shortest periodic orbits (infinitely many) are relevant for a periodic orbit analysis of spectral quantities in contrast to the generic case where the longest periodic orbits are essential.


The trace formula \eref{z2} is the first example of an exact trace formula for an infinite quantum graph. We have shown in \cite{SteinerEndres:2008,Egger:2011} that the infinite Dirichlet quantum graph can be obtained by truncating the graph $\mathfrak{G}$ at an arbitrary vertex $N$ and then consider the limit $N$ to infinity. Therefore, the corresponding trace formulae for the finite quantum graphs must converge to the summation formula \eref{z2}. It is easy to perform this limit on the l.h.s. of the trace formulae, i.e. on the quantum mechanical part. However, to perform the limit on the periodic orbit parts of the trace formulae for the truncated graphs
is a highly non-trivial problem since one has to cope with divergent terms, which explains why the derivation of the formula \eref{z2} along these lines is still an unsolved problem. It is therefore very interesting to find a proof of the Vorono\"i summation formula \eref{z2} interpreted as the trace formula for the Dirichlet quantum graph $\mathfrak{G}$ based on physical quantities only.

It is the purpose of this note to prove the Vorono\"i summation formula (\ref{13}) for a function space different from that in theorem \ref{13a} and e.g. in \cite{Dixon:1931,Wilton:1932,Dixon:1937,Hejhal:1979}. Furthermore, we want to use in the proof only quantities which are directly related to the Dirichlet quantum graph and which are examined in \cite{SteinerEndres:2008,Egger:2011} such as the trace of the Euclidean wave group of the Dirichlet Laplacian that we denote herein by $\Theta(t)$. In our proof we apply the Poisson summation formula, but here in the more general form of Dixon and Ferrar \cite{Dixon:1937}. In the concluding remarks we discuss how our function space for the Vorono\"i summation formula is related to the known function spaces of \cite{Dixon:1937} and \cite{Hejhal:1979}. 
\section{An alternative proof of the Vorono\"i summation formula}
\label{109}

In order to present an alternative proof of the Vorono\"i summation formula, we introduce the function (a particular Lambert series) 
\be
\label{110a}
\The(z):=\sum\limits_{n=1}^{\infty}d(n)\ue^{-nz}=\sum\limits_{n=1}^{\infty}\frac{1}{\ue^{nz}-1},\quad \re z>0,
\ee
which coincides for $z\equiv t\in\rz_{>0}$ with the physical quantity $\Thd(t)$ considered in \cite{SteinerEndres:2008,Egger:2011} and is the trace of the Euclidean wave group of the Dirichlet Laplacian on the infinite graph $\mathfrak{G}$ introduced in \cite{SteinerEndres:2008,Egger:2011}. In \cite{SteinerEndres:2008,Egger:2011} we derived the following asymptotic relations, $z\equiv t\in\rz_{>0}$,
\begin{eqnarray}
\label{24}
%
%
\The(t) &=&-\frac{\ln t}{t}+\frac{\gamma}{t}+\frac{1}{4}+\Or(t),\quad t\rightarrow 0^+,\\
\The(t) &=&\Or\left(\ue^{-t}\right), \quad  t\rightarrow\infty,\label{24a}
%
%
\end{eqnarray}
where the last relation  is trivial. Wigert \cite[p,\,203]{Wigert:1916} derived for $\Theta(z)$, $\re z>0$, a non-trivial identity, however, without introducing the digamma function $\psi(z):=\frac{\Gamma'(z)}{\Gamma(z)}$, by means of the Euler-MacLaurin summation formula, cited here for the special case $z\equiv t\in\rz_{>0}$ (A new proof of theorem \ref{121a} involving a Poisson summation formula will be presented below):
\begin{theorem}[Wigert:\,1916, \cite{Wigert:1916}]
\label{121a}
For $\Th(t)$, defined in $(\ref{110a})$, it holds the identity 
\be
\label{121}
\hspace{-2cm}\Th(t)=-\frac{\ln t}{t}+\frac{\gamma}{t}+\frac{1}{4}-\frac{2}{t}\sum\limits_{n=1}^{\infty}\left[\re\psi\left(1+\ui\frac{2\pi n}{t}\right)-\ln\left(\frac{2\pi n}{t}\right)\right],\quad  t>0.
\ee
\end{theorem}                                               
Our proof of the Vorono\"i summation formula will be based on a different, but equivalent representation of $\Theta(t)$ which follows from the following theorem:
\begin{theorem}
\label{125b}
It holds the identity , $t>0$,
\be
\label{125c}
\hspace{-2cm}\eq{
&\sum\limits_{n=1}^{\infty}\left[\re\psi\left(1+\ui\frac{2\pi n}{t}\right)-\ln\left(\frac{2\pi n}{t}\right)\right]\\
&\hspace{1.5cm}=\sum\limits_{n=1}^{\infty}\dz(n)\left[\exp\left(\frac{4\pi^2n}{t}\right)\Ei\left(-\frac{4\pi^2n}{t}\right)+\exp\left(-\frac{4\pi^2n}{t}\right)\Ei\left(\frac{4\pi^2n}{t}\right)\right]. 
}
\ee
\end{theorem}
In \eref{125c} $\Ei(x)$ is the exponential integral function defined for $x>0$ by the Cauchy principal value (see \cite[p.\,342]{Magnus:1966}), where $\Ei(x)\equiv\Es(x)$ and $\Ei(-x)\equiv-\Ee(x)$ for $x>0$. 

In order to prove theorems \ref{121a} and \ref{125b}, we require the following theorem:
%
\begin{theorem}[Dixon/Ferrar:\,1937, \cite{Dixon:1937}]  
\label{123}                                              
If
\begin{equation}
\label{124}
H(x):=h(x)-b\ln x,\quad x\in(0,\infty),
\end{equation}
is of bounded variation in the neighbourhood of $x=0$, then it holds, $\alpha\in(0,\infty)$,  
%
%
%
\be
\label{114}
\eq{
\sum\limits_{n=1}^{\infty}h(\al n) &= \frac{1}{2}b\ln(2\pi)-\frac{b}{2}\ln\alpha-\frac{1}{2}H\left(0^+\right)+\frac{1}{\al}\int\limits_0^{\infty}h(x)\ud x\\
& \hspace{2.5cm}+2\sum\limits_{n=1}^{\infty}\left[\frac{1}{\al}\int\limits_0^{\infty}h(x)\cos\left(\frac{2\pi n}{\al}x\right)\ud x+\frac{b}{4n}\right]
}
\ee
provided that
\begin{itemize}
\item the sum on the l.h.s. and the first integral on the r.h.s. in $\eref{114}$ exist,
\item $h(x)\rightarrow 0$ as $x\rightarrow \infty$,
\item $h(x)$ is the integral of $h'(x)$ in $x\geq x_0$,
\item $|h'(x)|$ is integrable over $(x_0,\infty)$.
\end{itemize}
\end{theorem} 
\hspace{-0.85cm}(Note that in \cite{Dixon:1937} the relation (\ref{114}) has been stated for the special value $\alpha\equiv1$).

We first present a new proof of theorem \ref{121a} which involves the Poisson summation formula (\ref{114}) rather than the Euler-MacLaurin summation formula.

\begin{proof}[New proof of theorem \ref{121a}]
We ``regularize'' $\The(t)$ in (\ref{110a}) at $t=0$:
\be
\The(t) &=& \sum\limits_{n=1}^{\infty}\frac{1}{\ue^{nt}-1}=\sum\limits_{n=1}^{\infty}\left[\frac{1}{\ue^{nt}-1}-\frac{\ue^{-nt}}{nt}\right]+\sum\limits_{n=1}^{\infty}\frac{\ue^{-nt}}{nt}\nonumber\\
       &=& \sum\limits_{n=1}^{\infty}h(nt)-\frac{\ln\left(1-\ue^{-t}\right)}{t},\quad t\in(0,\infty),\label{110}
\ee 
with
\be
\label{111}
h(x):=\frac{1}{\ue^x-1}-\frac{\ue^{-x}}{x}=\ue^{-x}\left[\frac{1}{1-\ue^{-x}}-\frac{1}{x}\right],\quad x\in (0,\infty).
\ee
We notice that
\be
\label{112}
h(x)=\frac{1}{2}+\Or(x),\quad x\rightarrow0^+, \quad \mbox{and}\quad h(x)=\Or\left(\ue^{-x}\right), \quad x\rightarrow\infty.
\ee
Now, we use theorem \ref{123} identifying $b=0$ for this case. With the identity \cite[p.\,16]{Magnus:1966} for the digamma function 
\be
\label{115}
\psi(z)=\ln z+\int\limits_{0}^{\infty}\left[\frac{1}{x}-\frac{1}{1-\ue^{-x}}\right]\ue^{-zx}\ud x,\quad \re z>0,
\ee
we obtain the relation ($n\in\nz_0$, $t\in(0,\infty)$)
\be
\label{116}
\int\limits_{0}^{\infty}h(x)\cos\left(\frac{2\pi n}{t}x\right)\ud x =\re\left\{\ln\left(1+\ui\frac{2\pi n}{t}\right)-\psi\left(1+\ui\frac{2\pi n}{t}\right)\right\}.
\ee
Furthermore, with ($n\in\nz$, $t\in(0,\infty)$)
\be
\label{117}
\re\ln\left(1+\ui\frac{2\pi n}{t}\right)=\ln\left(\frac{2\pi n}{t}\right)+\frac{1}{2}\ln\left(1+\left(\frac{t}{2\pi n}\right)^2\right)
\ee
and the identity \cite[p.\,85]{Abramowitz:1992}
\be
\label{118}
\sinh z= z\prod\limits_{n=1}^{\infty}\left(1+\left(\frac{z}{n\pi}\right)^2\right),\quad z\in\kz,
\ee
we obtain ($t\in(0,\infty)$)
\be
\label{119}
\sum\limits_{n=1}^{\infty}\ln\left[1+\left(\frac{t}{2\pi n}\right)^2\right]=\ln\left[\frac{\sinh\left(\frac{t}{2}\right)}{\left(\frac{t}{2}\right)}\right]=\frac{t}{2}+\ln\left(1-\ue^{-t}\right)-\ln t.
\ee
Using (see (\ref{116}))
\be
\label{a100}
h\left(0^+\right)=\frac{1}{2},\quad \int\limits_{0}^{\infty}h(x)\ud x=-\psi(1)=\gamma
\ee
and combining the Poisson summation formula (\ref{114}) with the results of (\ref{110}), (\ref{116}), (\ref{117}), (\ref{119}) and (\ref{a100}), we altogether have proved theorem \ref{121a}.
\end{proof}
\begin{proof}[Proof of theorem \ref{125b}]
Due to the asymptotics \cite[pp.\,346, 347]{Magnus:1966} (after correcting a typographic mistake by replacing $\ue^{-x}$ by $\ue^x$ on p. 347 l.c.)
\begin{equation}
\label{28}
e^x\Ei(-x)+e^{-x}\Ei(x)=\sum\limits_{m=0}^{M}\frac{2(2m+1)!}{x^{2m+2}}+\Or\left(\frac{1}{x^{2M+4}}\right),\quad x\rightarrow \infty,
\end{equation}
we can write the last sum in (\ref{125c}) as a double sum
\begin{equation}
\label{125a}
\sum\limits_{n=1}^{\infty}d(n)h\left(\frac{4\pi^2n}{t}\right)=\sum\limits_{m=1}^{\infty}\sum\limits_{n=1}^{\infty}h\left(\frac{4\pi^2m}{t}n\right),\quad t\in(0,\infty),
\end{equation}
where we have defined
\be
\label{126}
h(x):=e^x\Ei\left(-x\right)+e^{-x}\Ei\left(x\right),\quad x\in(0,\infty).
\ee
Notice that, due to \cite[p.\,343, 346, 347]{Magnus:1966} and (\ref{28}), the function $h(x)$ satisfies the required conditions of theorem \ref{123}. Furthermore, it holds \cite[p.\,343]{Magnus:1966}
\be
\label{127}
h(\al x)=2\gamma+2\ln\al+2\ln x +\Or(x^2\ln x), \quad x\rightarrow0^+,\quad \al\in(0,\infty),
\ee
and thus we identify in this case $b=2$. By Fourier's inversion formula we obtain with \cite[p.\,8]{Magnus:1954}
\be
\label{128}
\int\limits_{0}^{\infty}h(x)\cos\left(\frac{2\pi n}{\al}x\right)\ud x=-\frac{\frac{2\pi^2n}{\al}}{\left(\frac{2\pi n}{\al}\right)^2+1},\quad \al\in(0,\infty),\quad n\in\nz_0.
\ee
Therefore, we obtain by (\ref{114}) and (\ref{128}) setting $\al:=\frac{4\pi^2m}{t}$ ($m\in\nz$, $t\in(0,\infty)$)
\be
\hspace{-2cm}\sum\limits_{n=1}^{\infty}h\left(\frac{4\pi^2m}{t}n\right) &=& \ln (2\pi)-\gamma-\ln\left(\frac{4\pi^2m}{t}\right)+\sum\limits_{n=1}^{\infty}\left[\frac{-n}{n^2+\left(\frac{2\pi m}{t}\right)^2}+\frac{1}{n}\right]\nonumber\\
 &=& -\gamma-\ln\left(\frac{2\pi m}{t}\right)+\sum\limits_{n=1}^{\infty}\frac{\left(\frac{2\pi m}{t}\right)^2}{\left(n^2+\left(\frac{2\pi m}{t}\right)^2\right)n}.\label{129}
\ee
With the identity \cite[p.\,106]{Hansen:1975} ($x,y\in(0,\infty)$)
\be
\label{130}
\sum\limits_{k=1}^{\infty}\frac{1}{\left((kx)^2+y^2\right)k}=\frac{1}{2y^2}\left[\psi\left(1+\ui\frac{y}{x}\right)+\psi\left(1-\ui\frac{y}{x}\right)+2\gamma\right],
\ee
we get for the last sum on the r.h.s. in (\ref{129}) ($m\in\nz$, $t\in(0,\infty)$)
\be
\label{131}
\sum\limits_{n=1}^{\infty}\frac{\left(\frac{2\pi m}{t}\right)^2}{\left(n^2+\left(\frac{2\pi m}{t}\right)^2\right)n}=\re\psi\left(1+\ui\frac{2\pi m}{t}\right)+\gamma.
\ee
Putting (\ref{114}), (\ref{126}), (\ref{129}) and (\ref{131}) together, we gain (\ref{125c}). 
\end{proof}
Let us remark that \eref{125c} has also been derived in \cite{Egger:2011}, where we have given the following decomposition (however, by using the Vorono\"i summation formula \eref{z2}) 
\be
\label{131a} 
\Theta(t)=\Th^W(t)+\Tho(t), \quad t>0,
\ee
in a ``Weyl term'' and an ''oscillatory term'' defined as, $t>0$,
\be
\label{27}
\fl\eq{
\Th^W(t) &:=-\frac{\ln t}{t}+\frac{\gamma}{t}+ \frac{1}{4},\\
\Tho(t) &:= -\frac{2}{t}\sum\limits_{n=1}^{\infty}\dz(n)\left[\exp\left(\frac{4\pi^2n}{t}\right)\Ei\left(-\frac{4\pi^2n}{t}\right)+\exp\left(-\frac{4\pi^2n}{t}\right)\Ei\left(\frac{4\pi^2n}{t}\right)\right].
}
\ee

For the next steps, we need the notion of the Laplace transform in the sense of \cite{Doetsch:1937}.
Therefore, we define the  \textsl{L}-function space \cite[p.\,13]{Doetsch:1937}.
\begin{defn}
\label{132}
A real or complex valued function $f$ is an $L$-function iff
\begin{itemize}
	\item $f$ is defined at least for $t>0$,
	\item in each finite interval $0<T_1\leq t\leq T_2$ the function $f$ is Riemann integrable,
	\item the improper Riemann integral
	\be
	\label{133}
  \lim_{\epsilon\rightarrow0}{\int\limits_{\epsilon}^{T}}|f(t)|\ud t,\quad T\in(0,\infty) 
	\ee
	     exists,
	\item there exists a real or complex $s_0$ such that for some fixed $T>0$ the following improper Riemann integral
	\be
	\label{133a}
	\lim_{\omega\rightarrow \infty}{\int\limits_{T}^{\omega}}|f(t)|\ue^{-s_0t}\ud t
	\ee
	     exists.
\end{itemize}
\end{defn}
Now, let $f$ be an \textsl{L}-function. We consider an expression of the form
\be
\label{134}
\int\limits_{0}^{\infty}\Th(t)f(t)\ud t,
\ee
where $\Theta(t)$ is defined in (\ref{110a}). Due to (\ref{24}) the improper integral (\ref{134}) exists if
\be
\label{135}
f(t)\sim Bt^{\beta}L(t), \quad t\rightarrow0, \quad \beta>0, \quad B\in\kz\quad \mbox{arbitrary},
\ee
where $L(t)$ is a slowly increasing function at $t=0$ (for instance the absolute value of the logarithmic function) which is continuous and positive  possessing the property \cite[p.\,201]{Doetsch:1937}
\be
\label{138}
\frac{L(ut)}{L(t)}\rightarrow1, \quad t\rightarrow 0, \quad u>0\quad \mbox{arbitrary}.
\ee
For such a function one can show \cite[p.\,202]{Doetsch:1937}, \cite{Karamata:1930} (replacing $t$ by $\frac{1}{t}$) that
\be
\label{139}
t^{\epsilon}L(t)\rightarrow0, \quad t\rightarrow 0, \quad \epsilon>0\quad \mbox{arbitrary}.
\ee
Furthermore, we assume for further convenience that the function $f(t)$ 
is a so-called
regulated function (i.e. for every $t$ in the domain of definition both the left and right limits $f(t-)$ and $f(t+)$ exist [but must not be equal])
and
possesses compact support. Since
\be
\label{136}
{\Th}_{N}(t):=\sum\limits_{n=1}^{N}d(n)\ue^{-nt}\leq\Th(t) \quad \mbox{for all} \quad N\in\nz,\quad t\in(0,\infty),
\ee
the improper Riemann integral (\ref{134}) converges uniformly and absolutely with respect to the summation in $\Th(t)$. Furthermore, ${\Th}_{N}(t)$ converges locally uniformly on $(0,\infty)$ to $\Th(t)$. Thus, integration and summation can be interchanged in (\ref{134}) and we obtain
\be
\label{137}
\int\limits_{0}^{\infty}\Th(t)f(t)\ud t=\sum\limits_{n=1}^{\infty}d(n)\hata{f}(n),
\ee
where $\hata{f}$ denotes the Laplace transform of $f$ defined as
\be
\label{122a}
\hata{f}(s):=\int\limits_0^{\infty}\ue^{-st}f(t)\ud t.
\ee
In order to check that the sum on the r.h.s. in (\ref{137}) converges, we can use a theorem of \cite[p.\,202]{Doetsch:1937}: 
\begin{theorem}[(Doetsch, \cite{Doetsch:1937})]
\label{140}
Let $f$ be an $L$-function possessing the asymptotics
\be
\label{141}
f(t)\sim Bt^{\beta}L(t),\quad t\rightarrow0, \quad \beta>-1, \quad B\in\kz \quad \mbox{arbitrary},
\ee
where $L(t)$ is a regulated increasing function at $t = 0$ (and fulfills therefore $(\ref{138})$ and $(\ref{139}))$. Then for the Laplace transform $\hata{f}$ of $f$ the following asymptotic formula holds:
\be
\label{142}
\hata{f}(s)\sim B\frac{\Gamma(\beta+1)}{s^{\beta+1}}L\left(\frac{1}{s}\right),\quad s\rightarrow\infty, \quad s\in\rz_{>0}.
\ee
\end{theorem}
Therefore, the sum on the r.h.s. of (\ref{137}) converges absolutely due to (\ref{9a}). Now, we use theorem \ref{121a} and theorem \ref{125b} respectively (\ref{27}) and attain with (\ref{137}) ($\beta>0$ see (\ref{135}))
\be
& &\sum\limits_{n=1}^{\infty}d(n)\hata{f}(n)=\int\limits_{0}^{\infty}\left[\Thw(t)+\Tho(t)\right]f(t)\ud t\nonumber\\
                                     & &\hspace{0.3cm}=-\int\limits_{0}^{\infty}\frac{\ln t}{t}f(t)\ud t+\gamma\int\limits_{0}^{\infty}\frac{f(t)}{t}\ud t+\frac{1}{4}\int\limits_{0}^{\infty}f(t)\ud t+\int\limits_{0}^{\infty}f(t)\Tho(t)\ud t.\label{143}
\ee
Note that by the assumptions on the function $f$ the improper Riemann integrals on the r.h.s of (\ref{143}) are identical with the corresponding Lebesgue integrals. Using the identities \cite[p.\,149]{Magnus:1954} ($t\in(0,\infty)$)
\be
\label{144}
-\int\limits_{0}^{\infty}\left(\gamma+\ln k\right)\ue^{-kt}\ud k=\frac{\ln t}{t}, \quad \int\limits_{0}^{\infty}\ue^{-kt}\ud k=\frac{1}{t},
\ee
we obtain, by Fubini's theorem (for Lebesgue integrals) and with our previous assumption on the function $f$, the identity
\be
\hspace{-1.5cm}-\int\limits_{0}^{\infty}\frac{\ln t}{t}f(t)\ud t+\gamma\int\limits_{0}^{\infty}\frac{f(t)}{t}\ud t+\frac{1}{4}\int\limits_0^{\infty}f(t)\ud t
=\int\limits_{0}^{\infty}(\ln k +2\gamma)\hata{f}(k)\ud k+\frac{\tilde{f}(0)}{4}.\label{145}
\ee
Again, the Lebesgue integrals on the r.h.s. in (\ref{145}) are identical with the corresponding improper Riemann integrals. Now, we use the required condition of the compact support of the function $f$. Thus, the last integral on the r.h.s. in (\ref{143}) is in fact an integral over a finite interval and on this interval, due to (\ref{28}), the sum ($t\in(0,\infty)$)
\be
\fl{\Tho}_{N}(t):=-\frac{2}{t}\sum\limits_{n=1}^{N}\dz(n)\left[\exp\left(\frac{4\pi^2n}{t}\right)\Ei\left(-\frac{4\pi^2n}{t}\right)+\exp\left(-\frac{4\pi^2n}{t}\right)\Ei\left(\frac{4\pi^2n}{t}\right)\right] \label{146}
\ee
converges uniformly to $\Tho(t)$ for $N\rightarrow\infty$. Therefore, we can interchange integration and summation in the last term on the r.h.s of (\ref{143}). Finally, we can again use Fubini's theorem following the above argumentation, and obtain with the relation (see e.g. \cite[p.\,352]{Prudnikov:1986} and \cite[p.\,266]{Prudnikov:1986})
\be
\label{147a}
\hspace{-1.5cm}\eq{
& \int\limits_{0}^{\infty}\left[\exp\left(\frac{4\pi^2n}{t}\right)\Ei\left(-\frac{4\pi^2n}{t}\right)+\exp\left(-\frac{4\pi^2n}{t}\right)\Ei\left(\frac{4\pi^2n}{t}\right)\right]\frac{f(t)}{t}\ud t\\
&\hspace{3.5cm}=\int\limits_{0}^{\infty}\left[\pi Y_0\left(4\pi\sqrt{nk}\right)-2K_0\left(4\pi\sqrt{nk}\right)\right]\hata{f}(k)\ud k
}
\ee
our final theorem:
\begin{theorem}
\label{147}
Let $f$ be a regulated $L$-function possessing compact support and the asymptotics $(\ref{135})$. Then the following formula holds:
\be
\sum\limits_{n=1}^{\infty}d(n)\hata{f}(n) &=& \int\limits_{0}^{\infty}(\ln k +2\gamma)\hata{f}(k)\ud k +\frac{\hata{f}(0)}{4}\nonumber\\
                                         &\phantom{=}&\hspace{-1.5cm} +2\pi\sum\limits_{n=1}^{\infty}d(n)\int\limits_0^{\infty}\left[\frac{2}{\pi}K_0\left(4\pi\sqrt{nk}\right)-Y_0\left(4\pi\sqrt{nk}\right)\right]\hata{f}(k)\ud k, \label{148}                                                
\ee
where $\hata{f}$ denotes the Laplace transform of $f$.
\end{theorem}

Note that (\ref{148}) is identical to the Vorono\"i summation formula (\ref{13}) if $\tilde{f}(k)$ is renamed $f(k)$.
\section{Concluding remarks}
\label{b1}
The proof we have given in section \ref{109} is one of the shortest proofs of the Vorono\"i summation formula. There exists to our knowledge only one similarly short rigorous proof by Ferrar in \cite{Ferrar:1935,Ferrar:1937}. However, Ferrar used a completely different method which is based on the theory of Mellin transforms and on arguments of complex analysis as e.g. the residuum calculus. 

Our method is the very first one which uses the Laplace transformation to derive the Vorono\"i summation formula. The second crucial ingredient in our method is the Poisson summation formula in theorem \ref{123} which possesses a priori no connection to the Vorono\"i summation formula. Therefore, it is very interesting to provide a link between these two formulae, particularly because both formulae play a very important role in number theory and physics. 

Although the application of a Poisson-like formula (i.e. a Fourier series representation) for proving the Vorono\"i summation formula was first used by Landau \cite{Landau:1927}, we would like to remark that Landau's proof is different and also much longer than our proof presented in this paper. The first one who directly applied the Poisson summation formula in his proof was Hejhal \cite{Hejhal:1979} who applied it to a double sum and then separated terms which hamper the application
of the Poisson summation formula. Hejhal's proof is very sophisticated and totally different to ours.    

Finally, let us compare theorem \ref{147} with the corresponding results in \cite{Wilton:1932,Dixon:1937,Hejhal:1979}, which were derived by more involved and sophisticated methods. In particular, let us compare the function spaces for which the corresponding summations formulae are valid. We recall that the function space which was used in \cite{Hejhal:1979} (see theorem \ref{13a}) is the vector space $C_0^2(\rz)$ consisting of functions which are two times continuously differentiable and possess compact support. In \cite{Wilton:1932,Dixon:1937} the Vorono\"i summation formula is valid for the vector space $\mathcal{F}$ consisting of functions $f$ which possess the following properties: 
\begin{itemize}
 \item $f(k)$ is a real function and is of bounded variation in the interval $(0,k_0)$ for some $k_0>0$ (in \cite{Dixon:1937} it was assumed for convenience that $f(k)$ is continuous at $k=1,2,\ldots$. Furthermore, in \cite{Dixon:1937} there were discussed some modifications concerning the behaviour of $f(k)$ at zero [a logarithmic singularity was involved]).
 \item $ \left(V_{0^+}^x f(k)\right)\ln x\rightarrow 0$ as $x\rightarrow 0^+$, 
 \item for some positive $\kappa$, $k^{\frac{1}{2}+\kappa}f(k)\rightarrow 0$ as $k\rightarrow\infty$,
 \item $f(k)$ is the (indefinite) integral of $f'(k)$ in $k\geq k_0$,
 \item for some positive $\kappa$ 
\begin{equation}
\label{13b}
 \int\limits_{\kappa}^{\infty}k^{\frac{1}{2}+\kappa}\left|f'(k)\right|\ud k<\infty,
\end{equation}
\end{itemize}
where $V_{0^+}^x f(k)$ denotes the total variation of $f(k)$ in $(0,x)$ (see e.g.\cite{Boyarsky1997}).

We denote by $\mathfrak{F}$ the function space consisting of functions $f$ which are regulated $L$-functions possessing compact support and the asymptotics (\ref{135}). In order to compare the function space which is used in the summation formula (\ref{148}) with the above function spaces used in the summation formulae \cite{Wilton:1932,Dixon:1937,Hejhal:1979}, we have to consider the image $\widetilde{\mathfrak{F}}$ with respect to the Laplace transform of $\mathfrak{F}$
\be
\label{b2}
\widetilde{\mathfrak{F}}:=\left\{\hata{f};\quad f\in\mathfrak{F}\right\},
\ee
where $\tilde{f}$ denotes the Laplace transform of $f$. Obviously, $\widetilde{\mathfrak{F}}$ is a vector space. In order to investigate the regularity of the functions $\hata{f}\in\widetilde{\mathfrak{F}}$ we use a theorem of \cite[p.\,43]{Doetsch:1937}.
\begin{theorem}[(Doetsch, \cite{Doetsch:1937})]
Let $f$ be a regulated $L$-function and let the Laplace transform $\hata{f}$ of $f$ possess the half-plane $\re s>\beta$ as domain of definition (the half-plane where the integral in $(\ref{122a})$ converges [see $\cite[p.\,17]{Doetsch:1937}]$). Then the Laplace transform $\hata{f}$ of $f$ is a holomorphic function (in particular infinitely differentiable) in the domain $\re s>\beta$ and the derivatives of $\hata{f}$ are given by
\be
\label{b3}
\hata{f}^{(n)}(s)=(-1)^n\int\limits_{0}^{\infty}\ue^{-ts}t^nf(t)\ud t,\quad \re s>\beta,\quad n\in\nz_0.
\ee
\end{theorem}
Since the function space $\mathfrak{F}$ consists of functions possessing compact support and the asymptotics (\ref{135}), we immediately get the following corollary combining (\ref{135}), theorem \ref{140} and (\ref{b3}):
\begin{cor}
\label{b4}
Every function $\hata{f}\in\widetilde{\mathfrak{F}}$ is a regular function (holomorphic on the entire complex plane $\kz$). Furthermore, the derivatives $\hata{f}^{(n)}$ of $\hata{f}\in\widetilde{\mathfrak{F}}$ possess the asymptotics
\be
\label{b5}
\hata{f}^{(n)}(s)=\orr\left(\frac{1}{s^{n+1}}\right),\quad s\in\rz, \quad s  \rightarrow \infty,\quad n\in\nz_0.
\ee
\end{cor}
\hspace{-3.7mm}Therefore, we infer that $\widetilde{\mathfrak{F}}$ is a proper subspace of $\mathcal{F}$ which is used in \cite{Wilton:1932,Dixon:1937}. 

In order to compare the function space $\widetilde{\mathfrak{F}}$ with $C_0^2(\rz)$ which is used in \cite{Hejhal:1979}, we have to check whether it is possible for a function $\hata{f}\in\widetilde{\mathfrak{F}}$ to possess compact support, in particular whether there exists an $s_0>0$ with
\be
\label{b6}
\hata{f}(s)=0, \quad \mbox{for all}\quad s>s_0. 
\ee
But by corollary \ref{b4} and the identity theorem for holomorphic functions (see e.g. \cite{Knopp:1945}) we infer that every function $\hata{f}\in\widetilde{\mathfrak{F}}$ fulfilling (\ref{b6}) must be the zero function. One can show (see e.g. \cite[p.\,34,35]{Doetsch:1937}) that the zero function $\hata{f}\equiv0$ (on $\kz$) must originate from an $L$-function $f$ which is equal to zero except for a set of Lebesgue measure zero. Therefore, we deduce
\be
\label{b7}
\widetilde{\mathfrak{F}}\cap C_0^2(\rz)=\left\{\hata{f}\equiv0\quad \mbox{on} \quad \kz\right\}.
\ee
%
%
%
\newline
%

S.\ E.\ would like to thank the graduate school ``Analysis of complexity, information
and evolution'' of the Land Baden-Württemberg for the stipend which has enabled this paper.
\newline
\appendix
{\small
\bibliographystyle{unsrt}
\bibliography{litver3}}
\parindent0em
\end{document}